\documentclass{article}

\usepackage{amsthm,amsmath,graphicx,tikz,float,cite}
\usepackage{algorithm}
\usepackage{algpseudocode,amssymb,amsmath}
\newtheorem{theorem}{Theorem}
\newtheorem{cor}{Corollary}
\newtheorem{claim}{claim}
\newtheorem{definition}{definition}
\newtheorem{lemma}{lemma}

\begin{document}

\title{Algorithmic Aspects of $ 2 $-Secure Domination in Graphs}
\author{J. Pavan Kumar and P.Venkata Subba Reddy\\
	Department of Computer Science and Engineering \\ 
	National Institute of Technology \\ 
	Warangal, Telangana 506004, India. \\
	jp.nitw@gmail.com, pvsr@nitw.ac.in}
\date{}
\maketitle
\begin{abstract}
Let $G(V,E)$ be a simple, undirected and connected graph. 
A dominating set $S \subseteq V(G)$ is called a $2$-\textit{secure dominating set} ($2$-SDS) in $G$, if for every pair of distinct vertices $u_1,u_2 \in V(G)$ there exists a pair of distinct vertices $v_1,v_2 \in S$ such that $v_1 \in N[u_1]$, $v_2 \in N[u_2]$ and $(S \setminus \{v_1,v_2\}) \cup \{u_1,u_2 \}$ is a dominating set in $G$. The $2$\textit{-secure domination number} denoted  by $\gamma_{2s}(G)$, equals the minimum cardinality of a $2$-SDS in $G$. Given a graph $ G$ and a positive integer $ k,$ the $ 2 $-Secure Domination ($ 2 $-SDM) problem is to check whether $ G $ has a $ 2 $-secure dominating set of size at most $ k.$ It is known that $ 2 $-SDM is NP-complete for bipartite graphs. In this paper, we prove that the $ 2 $-SDM problem is NP-complete for planar graphs and doubly chordal graphs, a subclass of chordal graphs. We strengthen the NP-complete result for bipartite graphs, by proving this problem is NP-complete for some subclasses of bipartite graphs namely, star convex bipartite, comb convex bipartite graphs. We prove that $ 2 $-SDM is linear time solvable for bounded tree-width graphs. We also show that the $ 2 $-SDM is W[2]-hard even for split graphs. The Minimum $ 2 $-Secure Dominating Set (M2SDS) problem is to find a $ 2 $-secure dominating set of minimum size in the input graph. We propose a $ \Delta(G)+1 $ $ - $ approximation algorithm for M2SDS, where $ \Delta(G) $ is the maximum degree of the input graph $ G $ and prove that M2SDS cannot be approximated within $ (1 - \epsilon) \ln(\vert V \vert ) $ for any $ \epsilon > 0 $ unless $ NP \subseteq DTIME(\vert V \vert^{ O(\log \log \vert V \vert )}) $. 
Finally, we show that the M2SDS is APX-complete for graphs with $\Delta(G)=4.$
\end{abstract}

\section{Introduction}
\label{intro}
Let $G(V,E)$ be a simple, undirected and connected graph. For graph theoretic terminology we refer to \cite{west}. For a vertex $v \in V$, the \textit{open neighborhood} of $v$ in $G$ is $N_G(v)$= \{$u \in V: (u,v) \in E$\}, the \textit{closed neighborhood} of $v$ is defined as $N_G[v]=N_G(v) \cup \{v\}$. If $S \subseteq V$, then the open neighborhood of $S$ is the set $N_G(S) = \cup_{v \in S} N_G(v)$. The closed neighborhood of $S$ is $N_G[S] = S \cup N_G(S)$. Let $ S \subseteq V$. Then a vertex $ w \in V $ is called a \textit{private neighbor} of $ v $ with respect to $ S $ if $ N[w] \cap S = \{v\}. $ If further $ w \in V \setminus S$, then $ w $ is called an \textit{external private neighbor} (\textit{epn}) of $v$.  \par 
A subset $S$ of $V$ is a \textit{dominating set} in $G$ if for every $u \in V \setminus S$, there exists $v \in S$ such that $(u,v) \in E$. The \textit{domination number} of $G$ is the minimum cardinality of a dominating set in $G$ and is denoted by $\gamma(G)$.  A set $ D $ is a $ 2 $-\textit{dominating set} if every vertex in $ V \setminus D $ has at least $ 2 $ neighbors in $ D $. A dominating set $S \subseteq V$ is said to be a \textit{secure dominating set} (SDS) in $G$ if for every $u \in V \setminus S$ there exists $v \in S$ such that $(u,v) \in E$ and $(  S \setminus \{ v \})  \cup \{ u \} $ is a dominating set of $G$.  In this context, we say $ v $ $ S$-\textit{defends} $u $ or $ u $ is $ S $-\textit{defended by} $ v $. The minimum cardinality of a SDS in $G$ is called the \textit{secure domination number} of $G$ and is denoted by $\gamma_s(G)$. Suppose a \textit{guard} at a vertex $ v $ of the graph can deal with a problem in its closed neighborhood. A dominating set $ S $ is said to be \textit{secure dominating set} if an attack occurs at any vertex $ u $ of the graph can be $ S $-defended by some vertex $ v \in S $ in its closed neighborhood. However, suppose if two attacks simultaneously happen at any two vertices of the graph, then how to defend both the vertices is an interesting problem. A dominating set $S \subseteq V(G)$ is called a $2$-\textit{secure dominating set} ($2$-SDS) in $G$, if for every pair of distinct vertices $u_1,u_2 \in V(G)$ there exists a pair of distinct vertices $v_1,v_2 \in S$ such that $v_1 \in N[u_1]$, $v_2 \in N[u_2]$ and $(S \setminus \{v_1,v_2\}) \cup \{u_1,u_2 \}$ is a dominating set in $G$. The $2$\textit{-secure domination number} denoted  by $\gamma_{2s}(G)$, equals the minimum cardinality of a $2$-SDS in $G$ and any minimum $ 2 $-secure dominating set is referred as $ \gamma_{2s}$-set of $ G $. Given a graph $ G$ and a positive integer $ k,$ the $ 2 $-Secure Domination ($ 2 $-SDM) problem is to check whether $ G $ has a $ 2 $-secure dominating set of size at most $ k.$ The computational complexity of $ 2 $-SDM has been shown to be NP-complete for split graphs and bipartite graphs \cite{endm}. The Minimum $ 2 $-Secure Dominating Set (M2SDS) problem is to find a $ 2 $-secure dominating set of minimum size in the input graph. 
\paragraph{Preliminaries:} A vertex $u \in N[v]$ is a \textit{maximum neighbor} of $v$ in $G$ if $N[w] \subseteq N[u]$ holds for each $w \in N[v]$. A vertex $v \in V(G)$ is called \textit{doubly simplicial} if it is a simplicial vertex and it has a maximum neighbor in $ G $. An ordering $\{v_1, v_2, \ldots, v_n\}$ of the vertices of $V(G)$ is a \textit{doubly perfect elimination ordering} (DPEO) of $G$ if $v_i$ is a doubly simplicial vertex of the induced subgraph $G_i = G[\{v_i, v_{i+1}, \ldots, v_n\}]$ for every $i$, $1 \le i \le n$. A graph $G$ is \textit{doubly chordal} if and only if $ G $ has a DPEO \cite{doubly}. A \textit{tree} is an undirected graph in which any two vertices are connected by exactly one path. A \textit{star} is a tree $T=(A, F)$, where $A=\{a_0, a_1, \ldots, a_n\}$ and $F=\{(a_0,a_i) : 1 \le i \le n\}$. A \textit{comb} is a tree $T=(A,F)$, where $A=\{a_1, a_2, \ldots, a_{2n}\}$ and $F=\{(a_i,a_{i+1})$ $:$ $1 \le i \le n-1\}$ $\cup$ $\{(a_i,a_{n+i})$ $:$ $1 \le i \le n\}$. A bipartite graph $G (A, B, E)$ is called \textit{tree convex bipartite graph} if there is an associated tree $T=(A, F)$ such that for each vertex $b$ in $B$, its neighborhood $N_G(b)$ induces a subtree of $T$ \cite{jiang}. Further if $T$ is a star or comb, then $G$ is called as star convex bipartite or comb convex bipartite respectively. A \textit{vertex cover} of an undirected graph $ G(V, E) $ is a subset of vertices $ V^\prime \subseteq V$ such that if edge $ (u, v) \in E $, then either $ u \in V^\prime $ or $ v \in V^\prime$ or both. 
\section{Complexity Results}\label{complexity}	
It is known that $ 2 $-SDM problem is NP-complete even for bipartite graphs and split graphs \cite{endm}. We continue investigating its NP-completeness in other special graphs. In particular, we prove that it is also NP-complete for planar graphs and doubly chordal graphs.
\noindent To show that the $2$-SDM for planar graphs is NP-complete, we use Vertex Cover problem which is NP-complete even for planar graphs \cite{vcplanar}, and is defined as follows. \\[4pt]
\noindent \textbf{Vertex Cover Decision Problem (Vertex-Cover)} \\ [6pt]
\textit{Instance:} A simple, undirected planar graph $G$ and a positive integer $k$.\\
\textit{Question:} Does there exist a vertex cover of size at most $ k $ in $ G$?

\begin{theorem}
	$2$-SDM is NP-complete for planar graphs.
\end{theorem}
\begin{proof}
	Suppose a set $S \subseteq V$, such that $|S| \le k$ is given as a witness to a yes instance. It can be verified in polynomial time that $S$ is a $2$-SDS of $G$. Hence $2$-SDM is in NP.\par
	We reduce from Vertex-Cover problem to  $ 2 $-SDM for planar graphs.
	We claim that $G$ has a vertex cover of size at most $k$ if and only if $G^*$ has a $ 2 $-SDS of size at most $ r= m+ n+ k + 2$. 
\end{proof}
\noindent To show that the $2$-SDM is NP-complete for doubly chordal graphs, we use a well known NP-complete problem, called Exact Cover by 3-Sets (X3C) \cite{garey}, which is defined as follows. \\[4pt]
\textbf{Exact Cover By 3-Sets (X3C)}\\ [6pt]
\textit{Instance:} A finite set $ X $ with $\vert X \vert = 3q$ and a collection $ C $ of 3-element subsets\\ \hspace*{2.0cm}of $ X $. \\
\textit{Question:} Does $ C $ contain an exact cover for $ X $, that is, a sub collection $ C^\prime \subseteq C $\\ \hspace*{1.6cm} such that every
element in $ X $ occurs in exactly one member of $ C^\prime $?
\begin{theorem}
	$2$-SDM is NP-complete for doubly chordal graphs.
\end{theorem}
\begin{proof}
	It is known that the $ 2 $-SDM is a member of NP. To show that it is NP-complete, we propose a polynomial	time reduction from X3C. 
	We claim that the given instance of X3C $<X,$ $C>$ has an exact cover if and only if the constructed graph $G$ has a $ 2 $-SDS of size at most $l=q+2$. 
\end{proof}
\subsection{Complexity in some subclasses of bipartite graphs}
\noindent To prove the following theorem, we use a restricted version of Exact Cover by 3-Sets, which we denote by RX3C.\\[4pt]
\textbf{Restricted Exact Cover by 3-Sets (RX3C)}\\ [5pt]
{\textit{Instance:} A set $ X $ with $\vert X \vert = 3q$ and a collection $ C $ of 3-element subsets of $ X $\\ \hspace*{0.99cm}with $ \vert C \vert=m >3q$ and each element in $ X $ occurs in at most $ q $ subsets. \\
	\textit{Question:} Does $ C $ contain an exact cover for $ X $? }
\begin{theorem}
	$2$-SDM is NP-complete for star convex bipartite graphs.
\end{theorem}
\begin{proof}
	Clearly $2$-SDM is in NP. The proof is by reduction from RX3C problem. 
	We claim that RX3C instance $ <X,C> $ has a solution $ C^\prime $ if and only if $G$ has a $ 2 $-SDS of size at most $q+8$. 
\end{proof}
\begin{theorem}
	$2$-SDM is NP-complete for comb convex bipartite graphs.
\end{theorem}
\begin{proof}
	Clearly $2$-SDM is in NP. We transform an instance of X3C problem to an instance of $2$-SDM for comb convex bipartite graphs.
Next we show that X3C instance $ <X,C> $ has a solution $ C^\prime $ if and only if $G$ has a $ 2 $-SDS of size at most $q+8$.
\end{proof}
\subsection{Parameterized Complexity}
Now, we investigate the parameterized complexity of $ 2 $-SDM problem for split graphs. In \cite{endm}, $ 2 $-SDM has been proved as NP-complete for split graphs. 
\noindent The decision version of domination problem is defined as follows.\\[4pt]
\noindent 
\textbf{Dominating Set Decision Problem (DM)} \\ [6pt]
\textit{Instance:} A simple, undirected graph $G$ $(V, E)$ and a positive integer $k$.\\
\textit{Question:}  Does there exist a dominating set of size at most $ k $ in $ G $ ?\\[4pt]
\noindent In \cite{vraman}, the DM problem has been proved as W[2]-complete, even when restricted to split graphs.
\begin{theorem}
	$2$-SDM is W[2]-hard for split graphs.
\end{theorem}
\begin{proof}
	The proof is by reduction from DM problem. 
	We show that $G$ has a dominating set of size at most $k$ if and only if $G^*$ has a $ 2 $-SDS of size at most $ r= k + 2$. 
\end{proof}
\noindent Since split graphs form a proper subclass of chordal graphs, the following corollary is immediate.
\begin{cor}
	$2$-SDM is W[2]-hard for chordal graphs.
\end{cor}
\subsection{Complexity in bounded tree-width graphs}
Let $ G $ be a graph, $ T $ be a tree and $v\ $ be a family of vertex sets $ V_t \subseteq V (G) $ indexed by the vertices $ t $ of $ T $ . 
The pair $ (T, v\ ) $ is called a tree-decomposition of $ G $ if it satisfies the following three conditions:
(i) $ V(G)  =  \bigcup_{t \in V(T)} V_t $,
(ii) for every edge $ e \in E(G) $ there exists a $ t \in V(T)$ such that both ends of $ e $ lie in $ V_t $,
(iii) $V_{t_1} \cap   V_{t_3}  \subseteq V_{t_2}$ whenever $ t_1$, $t_2$, $t_3 \in V(T) $ and $ t_2 $ is on the path in $ T $ from $ t_1 $ to $ t_3 $.
The width of $ (T, v\ ) $ is the number $ max\{\vert V_t \vert-1 : t\in T \}$, and the tree-width $ tw(G) $ of $ G $ is the minimum width of any tree-decomposition of $ G $. By Courcelle's Thoerem, it is well known that every graph problem that can be described by counting monadic second-order logic (CMSOL) can be solved in linear-time in graphs of bounded tree-width, given a tree decomposition as input \cite{courc}. We show that $ 2 $-SDM problem can be expressed in CMSOL. 
\begin{theorem}[\textit{Courcelle's Theorem}](\cite{courc})\label{cmsol1}
	Let $ P $ be a graph property expressible in CMSOL and let $ k $ be
	a constant. Then, for any graph $ G $ of tree-width at most $ k $, it can be checked in
	linear-time whether $ G $ has property $ P $.
\end{theorem}
\begin{theorem}\label{cmsol2}
	Given a graph $ G $ and a positive integer $ k $, $ 2 $-SDM can be expressed in CMSOL.
\end{theorem}
\begin{proof}
	First, we present the CMSOL formula which expresses that the graph $ G $ has a dominating set of size at most $ k.$
	{\small 	$$Dominating(S)= (\vert S \vert \le k)  \land (\forall p)((\exists q)(q\in S \land adj(p,q))) \lor (p\in S)$$}
	where $ adj(p, q) $ is the binary adjacency relation which holds if and only if, $ p, q $ are two adjacent vertices of $ G.$
	$ Dominating(S) $ ensures that for every vertex $ p \in V $, either $ p\in S $ or $ p $ is adjacent to a vertex in $ S$ and the cardinality of $ S $ is at most $ k.$ 
	Now, by using the above CMSOL formula we can express $ 2 $-SDM in CMSOL formula as follows.
	\begin{center}
		{\small 	$ 2$-SDM($ S )=Dominating(S) \land (\forall x)(\forall y)((\exists p_1) (\exists p_2)(p_1 \in N[x] \land p_2 \in N[y]\setminus\{p_1\}$\\ \hspace*{1.9cm} $\land$ 	$Dominating((S \setminus \{x,y\}) \cup \{p_1,p_2\}))  $} 
	\end{center}

	\noindent Therefore, $ 2 $-SDM can be expressed in CMSOL.
\end{proof}
\noindent Now, the following result is immediate from Theorems \ref{cmsol1} and \ref{cmsol2}.
\begin{theorem}
	$ 2 $-SDM can be solvable in linear time for bounded tree-width graphs. 
\end{theorem}
\section{Approximation Results}
In this section, we obtain upper and lower bounds on the approximation ratio of the M2SDS problem. We also show that the M2SDS problem is APX-complete for graphs with maximum degree $ 4 $.
\subsection{Approximation Algorithm}
Here, we propose a $\Delta(G)+1$ approximation algorithm for the M2SDS problem. In this, we will make use of two known optimization problems, MINIMUM 2-DOMINATION and MINIMUM DOMINATION. The following two theorems are the approximation results which have been obtained for these two problems.
\begin{theorem}\label{app2dom}(\cite{ktuple})
	The MINIMUM k-TUPLE DOMINATION problem in a graph with maximum degree $ \Delta(G) $ can be approximated with an approximation ratio of $1+ \ln (\Delta(G)+1). $
\end{theorem}
\begin{theorem}(\cite{clrs})\label{appdom}
	The MINIMUM DOMINATION problem in a graph with maximum degree $ \Delta(G) $ can be approximated with an approximation ratio of $1+\ln (\Delta(G)+1).$
\end{theorem}

By Theorems \ref{app2dom} and \ref{appdom}, let us consider APPROX-2-DOM-SET and APP-ROX-DOM-SET are the approximation algorithms to approximate the solutions for MINIMUM 2-DOMINATION and MINIMUM DOMINATION with approximation ratios of $ 1+\ln (\Delta(G)+1)$ and $ 1+\ln (\Delta(G)+1)$ respectively.\par
Now, we propose an algorithm APPROX-2SDS to produce an approximate solution for the M2SDS problem. In APPROX-2SDS, first we compute 2-dominating set $ D_2 $ of a given graph $ G$ using APPROX-2-DOM-SET. Now let $ G^\prime = G[V \setminus D_2]$. By using APPROX-DOM-SET, we compute dominating set $ D^\prime $ of $ G^\prime$. Let $ D=D_2 \cup D^\prime.$ It can be easily observed that for any two vertices $ u_1,u_2 \in V $ there exist two vertices $ v_1 \in D \cap N[u_1] $ and $ v_2 \in D \cap N[u_2] $ such that $ (D \setminus \{v_1,v_2\}) \cup \{u_1,u_2\} $ is a dominating set of $ G.$ Therefore, $ D $ is a $ 2 $-SDS of $ G. $ 
\renewcommand{\algorithmicrequire}{\textbf{Input:}}
\renewcommand{\algorithmicensure}{\textbf{Output:}}
\begin{algorithm}
	\caption{APPROX-2SDS($ G $)}\label{2sdsalgo}
	\begin{algorithmic}[1]
		\Require{{A simple and undirected graph $G$ }}
		\Ensure{\mbox{A $ 2 $-SDS $ D $ of $G$.}}\\
		$ D_2  \gets$ \textbf{APPROX-2-DOM-SET} ($ G $) \\
		Let $ G^\prime= G[V\setminus D_2] $ \\
		$ D^\prime \gets$ \textbf{APPROX-DOM-SET} ($G^\prime$)  \\
		$ D \gets D_2 \cup D^\prime $ \\
		\Return $ D. $
		
	\end{algorithmic}
\end{algorithm}
\begin{theorem}\label{appscdom}
	The M2SDS problem in a graph $ G $ with maximum degree $ \Delta(G)$ can be approximated with an approximation ratio of $\Delta(G)+1.$ 
\end{theorem}
\begin{proof}
	To prove the theorem, we show that $ 2 $-SDS produced by our algorithm APPROX-2SDS, $ D$, is of size at most  $(\Delta(G)+1) $ times of $ \gamma_{2s}(G)$, i.e., $$ \vert D\vert  \le (\Delta(G)+1)\gamma_{2s}(G) $$  
	From the algorithm,
	\begin{center}
		$\vert D \vert = \vert D_2 \cup D^\prime \vert$ \\
		\hspace*{1.76cm}	$ = \vert D_2 \vert  + \vert D^\prime \vert \le n $\\
		\hspace*{1.95cm}  $ \le (\Delta(G)+1) \gamma(G) $  \\
		\hspace*{2.12cm}  $ \le (\Delta(G)+1) \gamma_{2s}(G)  $	
	\end{center}
\end{proof}
Since the M2SDS problem in a graph with maximum degree $ \Delta(G) $ admits an approximation algorithm that achieves the approximation ratio of $\Delta(G)+1$, we immediately have the following corollary of Theorem \ref{appscdom}.
\begin{cor}
	The M2SDS problem is in the class of APX when the maximum degree $ \Delta(G) $ is fixed.
\end{cor}
\subsection{Lower bound on approximation ratio}
To obtain a lower bound, we provide an approximation preserving reduction from the MINIMUM DOMINATION problem, which has the following lower bound.
\begin{theorem}\label{dsinapp}\cite{inapprox}
	For a graph $ G(V,E) $, the MINIMUM DOMINATION problem cannot be approximated within $ (1-\epsilon) \ln n$ for any $ \epsilon >0 $ unless NP $ \subseteq $ DTIME$( n ^{O(\log\log n)})$, where $ n=\vert V \vert$.  
\end{theorem}
\begin{theorem}\label{scdsinappthm}
	For a graph $ G(V,E) $, the M2SDS problem cannot be approximated within $ (1-\epsilon) \ln \vert V \vert $ for any $ \epsilon >0 $ unless NP $ \subseteq $ DTIME$( \vert V \vert ^{O(\log\log \vert V \vert)}). $ 
\end{theorem}
\begin{proof}
	In order to prove the theorem, we propose the following approximation preserving reduction. Let $ G(V,E) $, where $ V=\{v_1,v_2,\ldots,v_n\} $ be an instance of the MINIMUM DOMINATION problem. From this we construct an instance $ G^\prime(V^\prime, E^\prime) $ of M2SDS, where $ V^\prime=V \cup \{w_1, w_2, z_1, z_2, z_3\} $, and $ E^\prime=E \cup \{(v_i,w_1),$ $(v_i,w_2): v_i \in V\}  \cup \{(w_1,z_1) (w_2,z_2), (z_2, z_3)\}$.  \par
	Let $ D^*$ be a minimum dominating set of a graph $ G $ and $ S^*$ be a minimum $ 2 $-SDS of a graph $ G^\prime.$ It can be observed from the reduction that by using any dominating set of $ G, $ a $ 2 $-SDS of $ G^\prime $ can be formed by adding $ w_1,w_2$ and  $z_2$ vertices to it. Hence $ \vert S^* \vert \le \vert D^* \vert + 3. $ \par
	Let algorithm $ A $ be a polynomial time approximation algorithm to solve the M2SDS problem on graph $ G^\prime $ with an approximation ratio $ \alpha=(1-\epsilon) \ln \vert V^\prime \vert $ for some fixed $ \epsilon>0. $ Let $ k $ be a fixed positive integer. Next, we propose the following algorithm, DOM-SET-APPROX to find a dominating set of a given graph $ G $.
	
	\renewcommand{\algorithmicrequire}{\textbf{Input:}}
	\renewcommand{\algorithmicensure}{\textbf{Output:}}
	\begin{algorithm}
		
		\caption{DOM-SET-APPROX($ G $)}
		\begin{algorithmic}[1]
			\Require{\mbox{A simple and undirected graph $G$ }}
			\Ensure{\mbox{A dominating set $ D $ of $G$.}}
			\If {there exists a dominating set $ D^\prime $ of size at most $ k $} \\
			{\hskip1.5em $ D  \gets D^\prime$} 
			\Else \\
			{\hskip1.5emConstruct the graph $ G^\prime $ \\
				\hskip1.5emCompute a $ 2 $-SDS $ S $ of $ G^\prime $ by using algorithm $ A $ \\
				\hskip1.5em$ D\gets S \cap V $
				\hskip1.5em \If {$w_2 \in S$ and $ \exists v_u, v_u \in V\setminus N[D] $ }\\
				{\hskip3em  $D \gets D \cup \{v_u\}$}
				\EndIf
				\hskip1.5em \If {$w_1,z_1 \in S$ and $ \exists v_t, v_t \in V\setminus N[D] $ }\\
				{\hskip3em  $D \gets D \cup \{v_t\}$}
				\EndIf	}				  
			\EndIf \\
			\Return $ D. $
		\end{algorithmic}
	\end{algorithm}
	The algorithm DOM-SET-APPROX runs in polynomial time. It can be noted that if $ D $ is a minimum dominating set of size at most $ k $, then it is optimal. Next, we analyze the case where $ D $ is not a minimum dominating set of size at most $ k.$ \par
	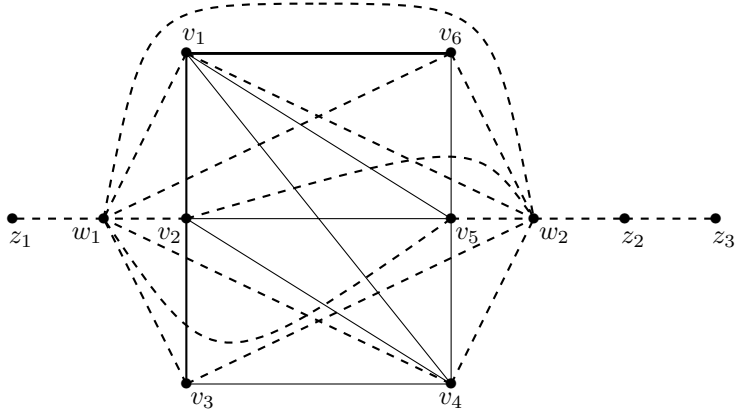
\begin{figure}[H]
		\begin{center}
			\begin{tikzpicture}[scale=1.1]
			\node at (0,0){\textbullet}; \node at (0.2,-0.2){$ v_3 $};
			\node at (0,2){\textbullet}; \node at (-0.2,1.8){$ v_2 $};
			\node at (0,4){\textbullet}; \node at (0.1,4.2){$ v_1 $};
			
			\node at (3.2,0){\textbullet}; \node at (3.2,-0.2){$ v_4 $};
			\node at (3.2,2){\textbullet}; \node at (3.4,1.8){$ v_5 $};
			\node at (3.2,4){\textbullet}; \node at (3.2,4.2){$ v_6 $};
			
			\draw[thick] (0,4)--(3.2,4)--(0,4);
			\draw[thick] (0,4)--(3.2,4)--(0,4);
			\draw[thick] (0,0)--(0,2)--(0,4);
			\draw (3.2,0)--(3.2,2)--(3.2,4);
			\draw (0,2)--(3.2,0);
			\draw (0,2)--(3.2,2)--(0,4)--(3.2,0);
			\draw (0,0)--(3.2,0);
			
			\node at (-1,2){\textbullet};\node at (-1.2,1.8){$ w_1 $};
			\node at (4.2,2){\textbullet};\node at (4.45,1.8){$ w_2 $};
			
			\draw[dashed,thick] (0,0)--(-1,2);
			\draw[dashed,thick] (3.2,0)--(-1,2);
			\draw[dashed,thick] (0,2)--(-1,2);
			\draw[dashed,thick] (3.2,2)..controls (0.35,0)..(-1,2);
			\draw[dashed,thick] (0,4)--(-1,2);
			\draw[dashed,thick] (3.2,4)--(-1,2);
			
			\draw[dashed,thick] (3.2,0)--(4.2,2);												
			\draw[dashed,thick] (0,0)--(4.2,2);												
			\draw[dashed,thick] (3.2,2)--(4.2,2);												
			\draw[dashed,thick] (0,2)..controls (3.5,3)..(4.2,2);												
			\draw[dashed,thick] (0,4)--(4.2,2);												
			\draw[dashed,thick] (3.2,4)--(4.2,2);

			\node at (-2.1,2){\textbullet};\node at (-2,1.75){$ z_1 $};
			\draw[dashed,thick] (-1,2)--(-2.1,2);
			
			\node at (5.3,2){\textbullet};\node at (5.4,1.75){$ z_2 $};
			\draw[dashed,thick]	 (4.2,2)--(5.3,2);
			
			\node at (6.4,2){\textbullet};\node at (6.5,1.75){$ z_3 $};
			\draw[dashed,thick]	 (5.3,2)--(6.4,2);
			
			\draw[dashed,thick] (-1,2)..controls (-0.5,4.6)..(1.6,4.6);
			\draw[dashed,thick]   (1.6,4.6)..controls (3.7,4.6)..(4.2,2);
			\end{tikzpicture}
			\caption{Example construction of a graph $G^\prime$}			
		\end{center}
	\end{figure}
	Let $ S^*$ be a minimum $ 2 $-SDS of $ G^\prime$, then $ \vert S^* \vert \ge k. $ Given a graph $ G $, DOM-SET-APPROX computes a dominating set of size $ \vert D \vert \le  \vert S \vert \le \alpha \vert S^* \vert \le \alpha (\vert D^* \vert + 3) = \alpha (1 + 3/\vert D^* \vert)\vert D^* \vert \le \alpha (1 + 3/k)\vert D^* \vert$. Therefore, DOM-SET-APPROX approximates a dominating set within a ratio $ \alpha (1 + 3/k).$ 
	If $ 3/k < \epsilon/2,$ then the approximation ratio $ \alpha (1 + 3/k) < (1-\epsilon) (1+\epsilon/2) \ln n= (1-\epsilon^\prime)\ln n$, where $\epsilon^\prime=\epsilon/2+\epsilon^2/2.$ \par
	\noindent By Theorem \ref{dsinapp}, if the MINIMUM DOMINATION problem can be approximated within a ratio of $ (1-\epsilon^\prime)\ln n,$ then $ NP \subseteq DTIME(n^{O(\log  \log n)})$. Similarly, if the M2SDS problem can be approximated within a ratio of $ (1-\epsilon)\ln n,$ then $ NP \subseteq DTIME(n^{O(\log  \log n)})$. For large values of $ n $, $ \ln n \approxeq \ln(n+5) $, for a graph $ G^\prime(V^\prime,E^\prime),$ where $ \vert V^\prime \vert = \vert V \vert +5,$ M2SDS problem cannot be approximated within a ratio of $ (1-\epsilon)\ln \vert V^\prime \vert$ unless $ NP \subseteq DTIME(\vert V ^\prime \vert^{O(\log \log \vert V^\prime \vert)} ).$
\end{proof}
\subsection{APX-completeness}
In this subsection, we prove that the M2SDS problem is APX-complete for graphs with maximum degree $ 4 $. This can be proved using an L-reduction, which is defined as follows.
\begin{definition}\textbf{(L-reduction)}
	Given two NP optimization problems $ F $ and $ G $ and a polynomial time transformation $ f $
	from instances of $ F $ to instances of $ G $, one can say that $ f $ is an \textit{L-reduction} if there exists positive constants $ \alpha $ and $ \beta $ such that for every instance $ x $ of $ F $
	\begin{enumerate}
		\item $ opt_G(f(x)) \le \alpha . opt_F(x) $.
		\item for every feasible solution $ y $ of $ f(x) $ with objective value $ m_G(f(x), y) = c_2 $ in polynomial time one can find a solution $ y^\prime $ of $ x $ with $ m_F (x, y^\prime) = c_1 $ such that $\vert  opt_F (x) - c_1 \vert \le  \beta \vert opt_G(f(x)) - c_2\vert. $
	\end{enumerate}
	Here, $ opt_F(x)$ represents the size of an optimal solution for an instance $ x $ of an NP optimization problem $ F$.  
\end{definition}
\noindent		An optimization problem $\pi $ is APX-complete if:
\begin{enumerate}
	\item $ \pi \in$ APX, and
	\item $ \pi \in$ APX-hard, i.e., there exists an L-reduction from some known APX-complete problem to $ \pi $.
\end{enumerate}
By Theorem \ref{appscdom}, it is known that the M2SDS problem can be approximated within a constant factor for graphs with maximum degree $ 4 $. Thus, M2SDS problem is in APX for graphs with maximum degree $ 4 $. To show APX-hardness of M2SDS, we give an L-reduction from MINIMUM DOMINATING SET problem in graphs with maximum degree $ 3$ (DOM-$ 3 $) which has been proved as APX-complete \cite{dom3apx}.
\begin{theorem}
	The M2SDS problem is APX-complete for graphs with maximum degree $ 4.$
\end{theorem} 
\begin{proof}
	It is known that M2SDS is in APX. Given an instance $ G(V,E)$ of DOM-$3$, where $ V=\{v_1,v_2,\ldots,v_n\} $, we construct an instance $ G^\prime(V^\prime, E^\prime) $ of M2SDS where $ V^\prime= V \cup \{x_i^1,x_i^2,x_i^3 : 1 \le i \le \lceil \frac{n}{2} \rceil\}$ and $ E^\prime=E \cup \{(v_i,x_{(i+1)/2}^1),$ $(v_{i+1},x_{(i+1)/2}^1) : 1 \le i \le n-1$ \& $i=1\pmod 2\} \cup \{(v_n,x_{(n+1)/2}^1) : n= 1 \pmod 2\} \cup \{(x_i^1,x_i^2),(x_i^2,x_i^3) : 1 \le i \le \lceil \frac{n}{2} \rceil\}.$ Note that $ G^\prime $ is a graph with maximum degree $ 4 $. An example construction of a graph $ G^\prime$ from a graph $ G$ is shown in Figure \ref{fig:APX}.  
	\begin{figure}
		\begin{center}
			\begin{tikzpicture}[scale=1.0]
			\draw[dotted] (1.25,2.9) ellipse (1.9cm and 2cm);
			\node at (1.2,4.5){$ G $};
			
			\node at (0,2){\textbullet}; \node at (0,1.7){$ v_2 $};			
			
			\node at (0,4){\textbullet}; \node at (-0.18,3.65){$ v_1 $};		
			\node at (-1.5,3){\textbullet}; \node at (-1.7,2.7){$ x_1^1 $};
			\node at (-3,3){\textbullet}; \node at (-3,2.7){$ x_1^2 $};
			\node at (-4.5,3){\textbullet}; \node at (-4.5,2.7){$ x_1^3 $};

			\node at (2.5,2){\textbullet}; \node at (2.5,1.7){$ v_3 $};
			
			\node at (4,3){\textbullet}; \node at (4.1,2.7){$ x_2^1 $};
			\node at (5.5,3){\textbullet}; \node at (5.5,2.7){$ x_2^2 $};
			\node at (2.5,4){\textbullet}; \node at (2.69,3.7){$ v_4 $};
			\node at (7,3){\textbullet}; \node at (7,2.7){$ x_2^3 $};
			
			\draw (2.5,4)--(2.5,2)--(0,2);
			\draw (2.5,4)--(0,4)--(0,2);
			\draw (0,2)--(2.5,4);

			\draw (0,4)--(-1.5,3)--(-3,3)--(-4.5,3);
			\draw (0,2)--(-1.5,3);
			\draw (2.5,2)--(4,3)--(5.5,3)--(7,3);
			\draw (2.5,4)--(4,3);
			
			\end{tikzpicture}
			\caption{Construction of $ G^\prime$ from $ G $}
			\label{fig:APX}
		\end{center}
	\end{figure}
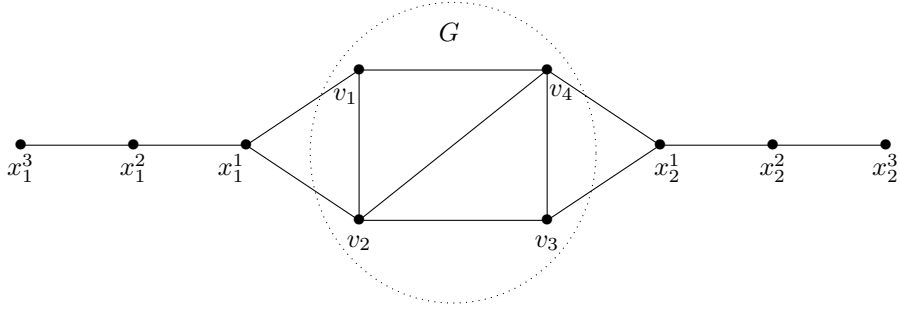
	\begin{claim}\label{claim1}
		If $ D^*$ is a minimum dominating set of $ G$ and $ S^*$ is a minimum $ 2 $-SDS of $ G^\prime$ then $ \vert S^* \vert = \vert D^* \vert + 2\lceil \frac{n}{2} \rceil,$ where $ n=\vert V \vert.$
	\end{claim}
	\begin{proof}[Proof of claim]
		Suppose $ D^*$ is a minimum dominating set of $ G$, then $ D^* \cup \{x_i^1,x_i^2 : 1 \le i \le \lceil \frac{n}{2} \rceil \} $ is a $ 2 $-SDS of $ G^\prime.$ Further, if $ S^*$ is a minimum $ 2 $-SDS of $ G^\prime $, then it is clear that $ \vert S^* \vert \le \vert D ^* \vert + 2 \lceil \frac{n}{2} \rceil.$ \par
		Next, we show that $ \vert S^* \vert \ge \vert D ^* \vert + 2 \lceil \frac{n}{2} \rceil.$ Let $ S $ be any $ 2 $-SDS of $ G^\prime$. It is clear that for any $i,$ where $1 \le i \le \lceil \frac{n}{2} \rceil, $ $\vert S \cap \{x_i^1,x_i^2,x_i^3\}\vert \ge 2$. Let $ D=S \cap V$ is not a dominating set of $ G.$ Then there exists a vertex $ v_j$ which is not dominated by $ D $ and consequently, two attacks simultaneously happen at vertices $ v_j, v_k^1 $ where $ v_k^1 \in N(v_j) \setminus V$ cannot be defended by $ S,$ which is a contradiction. Therefore, for every vertex $ u \in V \setminus D$ there exists a vertex $ v \in D $ such that $ (u,v) \in E $. Hence $D$ is a dominating set of $ G$ and $\vert D \vert \ge \vert D^* \vert$, which implies $ \vert S \vert \ge \vert D^* \vert+ 2 \lceil \frac{n}{2} \rceil.$ Since  $\vert S \vert \ge \vert S^* \vert$,  it is clear that $ \vert S^* \vert \ge \vert D ^* \vert + 2 \lceil \frac{n}{2} \rceil.$ 
	\end{proof}
	Let $ D^*$ and $ S^* $ be a minimum dominating set and a minimum $ 2 $-SDS of $ G $ and $ G^\prime$ respectively. It is known that for any graph $H$ with maximum degree $ \Delta(H) $, $ \gamma(H) \ge \frac{n}{\Delta(H)+1}$, where $ n=\vert V(H) \vert.$ Thus, $ \vert D^* \vert \ge \frac{n}{4}.$ From Claim \ref{claim1} it is evident that, $ \vert S^* \vert = \vert D ^* \vert + 2 \lceil \frac{n}{2} \rceil \le \vert D ^* \vert + n+1 \le \vert D ^* \vert+4\vert D ^* \vert+1 \le 6 \vert D ^* \vert.$ \par
	Now, consider a $ 2 $-SDS $ S $ of $ G^\prime$. Clearly, there exists a dominating set $ D$ in $ G$ of size at most $\vert S \vert - 2 \lceil \frac{n}{2} \rceil.$ Therefore, $ \vert D \vert  \le \vert S \vert - 2 \lceil \frac{n}{2} \rceil.$ Hence,  $ \vert D \vert-\vert D^* \vert  \le \vert S \vert - 2 \lceil \frac{n}{2} \rceil-\vert D^*\vert=\vert S \vert -\vert S^* \vert.$ This proves that there is an L-reduction with $ \alpha=6$ and $ \beta=1.$
\end{proof}
\section{Complexity difference in domination and $ 2 $-secure domination}
Although $ 2 $-secure domination is one of the several variants of domination problem, however they differ in computational complexity. In particular, there exist graph classes for which the first problem is polynomial-time solvable whereas the second problem is NP-complete and vice versa. Similar study has been performed between domination and other domination parameters in \cite{henning,jp}. \par

The DOMINATION problem is linear time solvable for doubly chordal graphs \cite{brandstadt}, but the $ 2$-SDM problem is NP-complete for this class of graphs which is proved in section \ref{complexity}. Now, we construct a class of graphs in which the $ 2$-SDM problem can be solved trivially, whereas the DOMINATION problem is NP-complete.
\begin{definition}\textbf{(GS graph)}
	A graph is \textit{GS graph} if it can be constructed from a connected graph $ G(V,E)$ where $ \vert V \vert=n,$ in the following way:\\[5pt]
	1. Create $ n $ star graphs $ \{S_1,S_2,\ldots,S_n\} $ each with $ 4 $ vertices, such that $ b_i $ as the central vertex and $ a_i,c_i,d_i $ as leaves of $ S_i $.\\[3pt]
	2. Attach graph $ G $ and $ S_i $ by joining $ v_i $ to $ a_i $, where $ 1 \le i \le n$.
\end{definition} 
\begin{theorem}
	If $ G^\prime$ is a GS graph obtained from a graph $ G(V,E) $ $ (\vert V \vert=n)$, then $ \gamma_{2s}(G^\prime)=3n$. 
\end{theorem}
\begin{proof}
	Let $ G^\prime(V^\prime, E^\prime) $ be a GS graph. An example construction of GS graph is illustrated in Figure \ref{gsgraph}. Let $ S= \{v_i, b_i, c_i : 1\le i\le n\}.$ It can be observed that $ S $ is a $ 2 $-SDS of $ G^\prime$ of size $ 3n$ and hence $ \gamma_{2s}(G^\prime) \le 3n.$ \par
	Let $ S^* $ be any $ \gamma_{2s} $-set in $ G^\prime $. It is clear that $\vert S^* \cap \{b_i,c_i,d_i\}\vert \ge 2,$ for each $ 1 \le i \le n$. If for some $ i $, $S^* \cap \{v_i,a_i\} = \emptyset $, then $\vert S^* \cap \{b_i,c_i,d_i\}\vert =3.$ 
	Thus, $ \vert  S^* \cap \{v_i,a_i,b_i,c_i,d_i\}\vert \ge 3$, for $ 1 \le i \le n.$ Hence $  \gamma_{2s}(G^\prime) \ge 3n. $	
\end{proof}
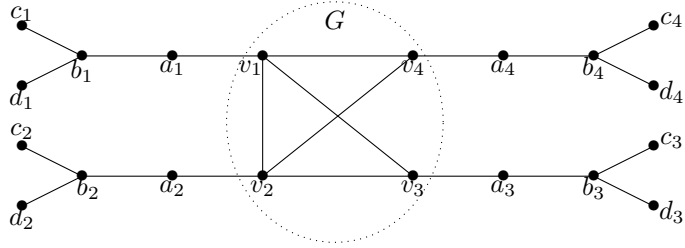
\begin{figure}[H]
	\begin{center}
		\begin{tikzpicture}[scale=0.8]
		\draw[dotted] (1.2,2.9) ellipse (1.8cm and 2cm);
		
		\node at (0,2){\textbullet}; \node at (0,1.8){$ v_2 $};			\node at (-1.5,2){\textbullet}; \node at (-1.5,1.8){$ a_2 $};
		\node at (-3,2){\textbullet}; \node at (-2.9,1.8){$ b_2 $};
		\node at (-4,2.5){\textbullet}; \node at (-4,2.75){$ c_2 $};
		\node at (-4,1.5){\textbullet}; \node at (-4,1.3){$ d_2 $};
		
		\node at (0,4){\textbullet}; \node at (-0.2,3.8){$ v_1 $};			\node at (-1.5,4){\textbullet}; \node at (-1.5,3.8){$ a_1 $};
		\node at (-3,4){\textbullet}; \node at (-3,3.75){$ b_1 $};
		\node at (-4,4.5){\textbullet}; \node at (-4,4.7){$ c_1 $};
		\node at (-4,3.5){\textbullet}; \node at (-4,3.3){$ d_1 $};
		
		\node at (4,2){\textbullet}; \node at (4,1.8){$ a_3 $};
		\node at (5.5,2){\textbullet}; \node at (5.5,1.8){$ b_3 $};
		\node at (6.5,2.5){\textbullet}; \node at (6.8,2.6){$ c_3 $};
		\node at (6.5,1.5){\textbullet}; \node at (6.8,1.4){$ d_3 $};
		\node at (2.5,2){\textbullet}; \node at (2.5,1.8){$ v_3 $};
		
		\node at (4,4){\textbullet}; \node at (4,3.8){$ a_4 $};
		\node at (5.5,4){\textbullet}; \node at (5.5,3.8){$ b_4 $};
		\node at (6.5,3.5){\textbullet}; \node at (6.8,3.4){$ d_4 $};
		\node at (6.5,4.5){\textbullet}; \node at (6.8,4.6){$ c_4 $};
		\node at (2.5,4){\textbullet}; \node at (2.5,3.8){$ v_4 $};
		
		\draw (0,4)--(2.5,2)--(0,2);
		\draw (2.5,4)--(0,4)--(0,2);
		\draw (0,2)--(2.5,4);

		\draw (0,4)--(-1.5,4)--(-3,4)--(-4,4.5);
		\draw (-3,4)--(-4,3.5);
		
		\draw (0,2)--(-1.5,2)--(-3,2)--(-4,2.5);
		\draw (-4,1.5)--(-3,2);
		\draw (2.5,2)--(4,2)--(5.5,2)--(6.5,2.5);
		\draw (6.5,1.5)--(5.5,2);
		\draw (2.5,4)--(4,4)--(5.5,4)--(6.5,4.5);
		\draw (6.5,3.5)--(5.5,4);	
		
		%
		
		
		%
		%
		
		
		\node at (1.2,4.6){$ G $};
		\end{tikzpicture}
		\caption{GS graph construction}
		\label{gsgraph}
	\end{center}
\end{figure}
\begin{lemma}\label{difflemma}
	Let $ G^\prime$ be a GS graph constructed from a graph $ G(V,E).$ Then $ G $ has a dominating set of size at most $ k $ if and only if $ G^\prime$ has a dominating set of size at most $ k+n.$
\end{lemma}
\begin{proof}
	Suppose $ D $ is a dominating set of $ G $ of size at most $ k,$ then it is clear that $ D \cup \{b_i : 1 \le i \le n\} $ is a dominating set of $ G^\prime $ of size at most $ k+n.$ 	Conversely, suppose $ D^\prime $ is a dominating set of $ G^\prime$ of size $ k+n.$ Clearly $ \vert D^\prime \cap \{a_i,b_i,c_i,d_i\} \vert \ge 1,$ for each $ 1 \le i \le n$. 
	Let $ D^{\prime\prime}$ be the set formed by replacing all $ a_i$'s in $ D^\prime $ with corresponding $ v_i $'s. Clearly, $ D^{\prime\prime} $ is a dominating set of $ G $ and $ \vert D^{\prime\prime} \vert \le k $.
\end{proof}
\noindent The following result is well known for the DOMINATION problem.
\begin{theorem}(\cite{garey})\label{dsnpc}
	The DOMINATION problem is NP-complete for general graphs.
\end{theorem}
\begin{theorem}
	The DOMINATION problem is NP-complete for GS graphs.
\end{theorem}
\begin{proof}
	The proof directly follows from Theorem \ref{dsnpc} and Lemma \ref{difflemma}.
\end{proof}
\noindent It is identified that the two problems, DOMINATION and $2$-SDM are not equivalent in computational complexity aspects. For example, when the input graph is either doubly chordal or a GS graph then complexities differ. Thus, there is a scope to study each of these problems on its own for particular graph classes.
\section{Conclusion}
In this paper, we have proved the NP-completeness of $ 2 $-SDM for planar graphs, doubly chordal graphs, star convex bipartite and comb convex bipartite graphs. On
the positive side, we have proved that a minimum cardinality $ 2 $-secure dominating
set of a graph with bounded tree-width can be computed in linear time. From approximation point of view, we have proposed an approximation algorithm for obtaining $ 2 $-SDS for general graphs. On the other side, we have also proved some approximation hardness results. It would be interesting to study the complexity of this problem in other graph classes such as interval graphs and block graphs.

\end{document}